\documentclass[10pt,reqno]{amsart}

\usepackage{amsmath}
\usepackage{amssymb}
\usepackage{amsfonts}
\usepackage{epsfig}

\newtheorem{theorem}{Theorem}[section]
\newtheorem{lemma}{Lemma}[section]

\newtheorem{remark}{Remark}[section]

\newtheorem{example}{Example}[section]
\numberwithin{equation}{section}

\begin{document}
\title{Optimal Constrained Investment in the Cramer-Lundberg model}
\author{Tatiana Belkina}\thanks{This research was supported by the Russian Fund of Basic Research, Grants RFBR 10-01-00767 and  RFBR 11-01-00219}
\address{Laboratory of Risk Theory, Central Economics and Mathematics Institute  of  the Russian Academy of Sciences. , Moscow, Russia}
\author{Christian Hipp}
\address{Institute for Finance, Banking and Insurance, Karlsruhe Institute of
Technology, Karlsruhe, Germany}
\author{Shangzhen Luo}
\thanks{This research was supported by a UNI summer fellowship}
\address{Department of Mathematics, University of Northern Iowa, Cedar Falls, IA, USA 50614}
\author{Michael Taksar}
\thanks{This research was supported by the Norwegian Research Council
Forskerprosjekt ES445026 ``Stochastic Dynamics of Financial Markets"}
\address{Department of Mathematics, University of Missouri, Columbia, MO, USA 60211}
\thanks{Correspondence Email: luos@uni.edu, Postal Address: Department of Mathematics, University of Northern Iowa, Cedar Falls, Iowa 50614-0506, USA}

\begin{abstract}
We consider an insurance company whose surplus is represented by
the classical Cramer-Lundberg process. The company can invest its
surplus in a risk free asset and in a  risky asset, governed by
the Black-Scholes equation. There is a constraint that the
insurance company can only invest in the risky asset at a limited
leveraging level; more precisely, when purchasing, the ratio of
the investment amount in the risky asset to the surplus level is
no more than $a$; and when shortselling, the proportion of the
proceeds from the short-selling
 to the surplus level is
no more than $b$. The objective is to find an optimal investment policy
that minimizes the probability of ruin.
The minimal ruin probability as a function of the initial surplus is
 characterized by a classical solution to
the corresponding Hamilton-Jacobi-Bellman (HJB) equation. We study
the optimal control policy and its properties. The interrelation
between the parameters of the model plays a crucial role in the
qualitative behavior of the optimal policy. E.g., for some ratios
between $a$ and $b$,  quite unusual and at first ostensibly
counterintuitive policies may appear, like short-selling a stock
with a higher rate of return to earn lower interest, or borrowing
at a higher rate to invest in a  stock with lower rate of return.
This is in sharp contrast with the unrestricted case, first
studied in Hipp and Plum (2000), or with the case of no
shortselling and no borrowing studied in Azcue and Muler (2009).
\end{abstract}
\maketitle

\section{Introduction}
Ruin minimization has become a classical criterion in optimization
models and in recent years it has been extensively studied, being
of a natural interest to the policymakers and supervisory
authorities of the insurance companies. Browne \cite{B} was one of
the first to consider the ruin minimization problem for a
diffusion model; there it was found that the optimal investment
policy is to keep a constant amount of money in the risky asset.
Taksar and Markussen in \cite{TM} studied a diffusion
approximation model in which one controls  proportional
reinsurance; they obtained the minimal ruin probability function
and the  optimal reinsurance policy in a closed form. In \cite{S}
Schmidli studied the ruin optimization problem for the classical
Cramer-Lundberg process, with the control  of proportional
reinsurance. In \cite{S1} a similar problem was considered with
investment and reinsurance control. 

In this paper we study a ruin probability minimization problem in
which there are constraints on the investment possibilities.
Namely, the insurance company has an opportunity to invest in a
financial market that consists of a risk free asset and a risky
asset, however, it can buy the risky asset up to the limit, which
is $a$ times the current surplus and it can shortsell the risky
asset up to the limit of no more than $b$ times the current
surplus. If $a>1$, then borrowing to invest is allowed and the
amount borrowed to buy the risky asset is no more than $(a-1)$
times the current surplus. The risky asset is governed by a
geometric Brownian motion, and the surplus is modeled by the
classical compound Poisson risk process. The objective is to find
the optimal investment policy which minimizes the ruin
probability. The model was first considered in Hipp and Plum
\cite{HP}, where there were no constraints on the investment
possibilities, that is the investment in the risky asset could be
at any amount (positive or negative) irrespective of the surplus
level. Further,  Azcue and Muler \cite{AM}, studied the same model
with no shortselling and no borrowing requirement.

It has been noticed, that in the case of no constraints on the investment
 (See \cite{B}, \cite{HP}, \cite{HP2} and \cite{PY}),
 the optimal investment strategy is highly leveraged
when the  surplus levels are small. There  are several papers in which
there are direct or indirect constraints imposed on the leveraging level.
In \cite{AM} bounds on the leveraging level are the result of
a no-shortselling-no-borrowing constraint.
In \cite{LUO}  only a limited amount of borrowing is allowed and it is at a higher rate than
the risk free lending/saving one.

In this paper we consider general constraints on borrowing and
shortselling, which are formulated in proportions to the surplus;
they can be higher or lower than those of
no-shortselling-no-borrowing ones. Note that in many recent
papers, e.g. \cite{AM}, \cite{HP}, and \cite{S}, the risk free
interest rate $r$ is assumed to be zero after inflation
adjustment, and the rate of return of the risky asset $\mu$ is
positive. It thus  excludes the case of $\mu<r$. Here, we assume a
positive interest rate $r$  and we do not assume any relationship
between $\mu$ and $r$. In \cite{HP}, the condition $\mu>r$  is
implicit (there $r=0$), and the optimal policy does not involve
short-selling even though it is allowed. The same phenomenon is
observed in diffusion approximation models as well, e.g.
\cite{LUO} and \cite{PY}.

The generality of constraints on investment brings a whole new
dimension to the possible qualitative behavior of the optimal
policies. Some of those might look counterintuitive, at first. For
example, depending on the relationship between $a$, $b$ and other
parameters of the model, the optimal policy might involve
shortselling not only when $\mu<r$ but also when $\mu>r$.
Moreover, the optimal policy might consist of switching from
maximal borrowing  to maximal shortselling then to maximal
borrowing again as the surplus level increases (see a detailed
analysis of the exponential claim size distribution case at the
end of this paper). This is a manifestation of a rather complex
interplay between the potential profit and risk at
different levels of surplus. In short, at some surplus levels it
is optimal for the company to leverage its risky or the risk free
asset (purchase or short-sell) at the maximum levels ($a$ or $b$);
and to bet on stock's volatility to increase the chances for the
surplus level to bounce back.

Our mathematical technique is based on operator theories
applied for the solution of the corresponding  HJB equation
 (see \cite{AM}, \cite{HP} and \cite{S}).
We first show existence of a classical solution to the
corresponding HJB equation. To this end we first observe that the minimizer
in the HJB equation is constant, on one of the edges of the admissible values for
the proportion, as long as the wealth of the insurer is small, $x<\varepsilon,$ say. 
Then we define a special operator $T$ in the space of continuous functions $f(x)$ on compact 
intervals $I=[\varepsilon, K]$ such that the HJB equation on $I$ is equivalent to $f'(x)=Tf(x)$. 
We prove that the operator $T$ is Lipschitz and conclude that the equation $f'(x)=Tf(x)$ has a solution
on $I.$ Finally we show that the solution, extended to $(0,\infty),$ is bounded and proportional to
the minimal ruin probability. 

The rest of the paper is organized as follows. The optimization
problem is formulated in Section 2. In Section 3,
an operator is defined to show existence of a classical solution to the HJB equation.
A verification theorem is proved in Section 4. In
Section 5 we investigate the case with exponential claim size
distribution and present several numerical examples. The last section is
devoted to economic analysis.

\section{The Optimization Problem}
We assume that without investment the surplus of the insurance company is
governed by the Cramer-Lundberg process:
$$X_t=x + ct-\underset{i=1}{\overset{N(t)}{\sum}}Y_i,$$
where $x$ is the initial surplus,
 $c$ is the premium rate, $N(t)$ is a Poisson process with intensity $\lambda$, the random variables $Y_i$'s are positive i.i.d. representing
 the size of the claims.
Suppose that at the time $t$, the insurance company invests a
fraction  $\theta_t$ of its surplus  into a risky asset whose
price follows a geometric Brownian motion
$$dS_t=\mu S_tdt+\sigma S_tdB_t.$$
Here $\mu$ is the stock return rate, $\sigma$ is the volatility,
and $B_t$ is a standard Brownian motion independent of
$\{N(t)\}_{t\geq 0}$ and $Y_i$'s. Then the fraction $(1-\theta_t)$
of the  the surplus is invested in the risk free asset whose price
is governed by
$$dP_t=rP_tdt,$$
where $r$ is the risk-free interest rate. If $0\le \theta_t\le 1$,
then the insurance company purchases the risky asset at a cost of
no more than its current surplus; if $\theta_t>1$, the insurance
company borrows to invest in the risky asset; and if $\theta_t<0$,
the insurance company shortsells the risky asset to invest in the
risk free asset. Let $\pi:=\{\theta_s\}_{s\geq 0}$ stand for the
control functional. Once $\pi$ is chosen, the surplus process
$X^\pi_t$ is governed by the equation below
\begin{equation}\label{dyn}
X^\pi_t=x+\int_0^t[c+r(1-\theta_s)X_s^\pi+\mu \theta_s
X^\pi_s]ds+\sigma \int_0^t\theta_s X^\pi_sdB_s
-\underset{i=1}{\overset{N(t)}{\sum}}Y_i,
\end{equation}

We assume all the random variables are defined on a complete probability
space $(\Omega, \mathcal{F}, P)$. On this space we define the filtration
 $\{\mathcal{F}_t\}_{t\geq 0}$
generated by processes $\{X_t\}_{t\geq 0}$ and $\{B_t\}_{t\geq
0}$. A control {\it policy} (or just a {\it  control})  $\pi$ is
said to be {\it admissible} if $\theta_t$ is
$\mathcal{F}_t$-predictable and it satisfies  $\theta_t\in
\mathcal {U}=[-b,a]$. We denote by $\Pi$ the set of all admissible
controls.

 In this paper, we make the following assumptions:
(i) the exogenous parameters $a$, $b$, $c$, $r$, $\mu$, $\sigma$,
are positive constants ($b=0$ is not allowed); (ii) the claim
distribution function $F$ has a finite mean and it has a
continuous density with support $(0,\infty)$.

The ruin time of the process $X_t^{\pi}$ under the investment strategy $\pi$ is defined as follows
\begin{eqnarray}
\label{ruintime}
\tau^\pi=\inf\{t\geq 0: X^\pi_t<0\},
\end{eqnarray}
and the survival probability as
\begin{eqnarray}
\label{svprob}
\delta^\pi(x)=1-P(\tau^\pi<\infty).
\end{eqnarray}
The maximal survival probability is defined as
\begin{eqnarray}
\label{opt-svp}
\delta(x)=\underset{\pi\in\Pi}{\sup}\ \delta^\pi(x),
\end{eqnarray}
which is a non-decreasing function of $x$. If we assume that
$\delta$ is twice continuously differentiable, then it solves the
following Hamilton-Jacobi-Bellman equation:
\begin{eqnarray}
\label{HJB} 
\underset{\theta
\in \mathcal{U}}{\sup}\mathcal{L}^{(\theta)}\delta(x)=0,x\ge 0,
\end{eqnarray}
where
\begin{equation}
\label{LM}
\begin{split}
&\mathcal{L}^{(\theta)}\delta(x)=\frac{\sigma^2x^2\theta^2}{2}\delta''(x)
+[c+rx+(\mu-r)\theta x]\delta'(x)-M(\delta)(x),\\
&M(\delta)(x)=\lambda[\delta(x)-\int_0^x\delta(x-s)dF(s)].
\end{split}
\end{equation}
We note that $M(\delta)(x)$ is positive,  given that
$\delta$ is an increasing function on $(0,\infty)$.
The minimum in the HJB equation is attained at some value $ \theta^*(x)\in \mathcal{U}$
and this means that
$$\mathcal{L}^{(\theta^*(x))}\delta(x)=0,x\ge 0,$$
while for all other values $\theta \in \mathcal{U}$ and $x\ge 0$
$$\mathcal{L}^{(\theta)}\delta(x)\le0.$$
This implies that for $x>0$
\begin{equation}\label{Good Operator}
\delta''(x)=2\inf_{\theta \in \mathcal{U}}\{ M(\delta)(x)-[c+rx+(\mu-r)\theta x]\delta'(x)\}/(\sigma^2\theta^2x^2).
\end{equation}
We have used this formula for numerical calculations when $x$ is large enough.

\section{Existence of a smooth solution to the HJB equation}
For any twice continuously differentiable
function $W$, let
\begin{eqnarray}
\label{alphaw}
\alpha_W(x)=-\frac{(\mu-r)W'(x)}{\sigma^2xW''(x)},
\end{eqnarray}
if $W''(x)\neq 0$.
Suppose $W$ is non-decreasing and solves HJB equation~\eqref{HJB} at $x$, then
we can define a maximizer in the following form
\begin{eqnarray}\label{alphastar}
\alpha^*_W(x)=\begin{cases}
\alpha_W(x)&\ \ W''(x)<0,\ -b\leq\alpha_W(x)\leq a;\\
a&\ \ W''(x)<0,\ \alpha_W(x)>a;\\
 &\ \ or\ \ W''(x)>0,\ \alpha_W(x)\leq\frac{a-b}{2};\\
 &\ \ or\ \ W''(x)=0,\ \mu\ge r;\\
-b&\ \ W''(x)<0,\ \alpha_W(x)<-b;\\
 &\ \ or\ \ W''(x)>0,\ \alpha_W(x)>\frac{a-b}{2};\\
 &\ \ or\ \ W''(x)=0,\ \mu<r,
\end{cases}
\end{eqnarray}
and we have
$$\mathcal{L}^{(\alpha^*_W(x))}W(x)=0,x>0.$$
So the minimum is always attained at one of the three points $-b$, $a$, or $\alpha_W(x).$

For $0\neq \gamma \in \mathcal{U}$ consider the equation
\begin{equation}\label{HJBgamma}
\mathcal{L}^{\gamma}V(x)=0,x>0.
\end{equation}
From Proposition 4.2 in \cite{AM}, p. 30, we obtain the existence of a function $V_{\gamma}(x),x>0,$
which is twice continuously differentiable on $(0,\infty),$ with 
\begin{eqnarray}\label{Vgamma}
V_{\gamma}(0+)&=&1,\\
V'_{\gamma}(0+)&=&\lambda/c, \\
\label{V'V''0}
V''_{\gamma}(0+)&=&\frac{\lambda}c \left( \frac{\lambda}c -F'(0+)-\frac{r+\gamma(\mu-r)}c \right)
\end{eqnarray}
satisfying equation (\ref{HJBgamma}). 
One can show the formula for $V''(0+)$ by a similar method of Proposition 4.2 in \cite{AM}.

In the following two Lemmas we show that for small values 
of $x$ the function $V_a(x)$ is a solution to the HJB equation if $\mu>r$, and
$V_{-b}(x)$ is a solution if $\mu<r.$

\begin{lemma}
\label{VVa}
For $\mu>r$, there exists $\varepsilon>0$
such that on $(0,\varepsilon)$ function $V_{a}$
solves
\begin{eqnarray}\label{HJBa}
\underset{\theta\in[-b,a]}{\sup}\mathcal{L}^{(\theta)}\delta(x)=\mathcal{L}^{(a)}\delta(x)=0.
\end{eqnarray}
\end{lemma}
\begin{proof}
We consider three cases separately and use the representation of the maximizer given in (\ref{alphastar}).
Firstly, for $x$ with $V_a''(x)=0$, $\alpha_{V_a}^*(x)=a$ is a maximizer if $\mu-r>0.$ 
Secondly, for $x$ small with $V_a''(x)>0$, since $V_a'(x)\ge 0$ for $x$ near $0,$ the relevant part for the maximum
$$\theta x(\mu-r)V_a'(x)+\frac 12 \theta^2\sigma^2x^2V''(x)$$
is increasing in $\theta,$ and so we again obtain $\alpha_{V_a}^*(x)=a.$ 
Thirdly for $x$ small with $V_a''(x)<0$,  it holds that
\begin{eqnarray*}
\alpha_{V_{a}}(x)=-\frac{(\mu-r)V'_a(x)}{xV''_a(x)}>\frac{a-b}{2};
\end{eqnarray*}
hence $\alpha_{V_a}^*(x)=a.$ So (\ref{HJBa}) holds in all three cases.
\end{proof}

\begin{lemma}
\label{VVb}
For $\mu<r$, there exists $\varepsilon>0$
such that on $(0,\varepsilon)$ function $V_{-b}$
solves
\begin{eqnarray}\label{HJBb}
\underset{\theta\in[-b,a]}{\sup}\mathcal{L}^{(\theta)}\delta(x)=\mathcal{L}^{(-b)}\delta(x)=0.
\end{eqnarray}
\end{lemma}

The results in lemmas \ref{VVa} and \ref{VVb} are intuitively appealing:
at low surplus levels,
when the stock return rate $\mu$ is higher than the interest rate $r$,
the company invests in the stock at the maximum level $a$;
on the other hand, when the interest rate is higher,
the company would invest in the risk-free asset at the maximum level of $1+b$
(all the surplus together with shortselling proceeds at the maximum level $b$).
 
For $\varepsilon >0$ the function $V_{\gamma}(x)$ with $\gamma \in \{a,-b\}$ 
is now extended to the range $[0,\infty)$ as follows.
To this end,
fix $\varepsilon<K<\infty,$ and a positive decreasing function $A(x)$ with $A(x)<\min(a,b)$ 
(to be chosen later). 
On the set $\mathcal{C}$ of functions $w(x)$ which are continuous on $[\varepsilon,K]$ 
consider the operator
\begin{equation}\label{Operator} 
Tw(x)=2\inf_{\theta \in \mathcal{U},|\theta| > A(x)} \{ M_{\gamma}(W)(x) -[c+rx+\theta x(\mu-r)]w(x)\} /(\theta^2\sigma^2x^2),
\end{equation}
where $$M_{\gamma}(W)(x)=\lambda \left(W(x)-\int_0^{\varepsilon} V_{\gamma}(y)f(x-y)dy-\int_{\varepsilon}^x W(y)f(x-y)dy\right),$$ 
$$W(x)=V_{\gamma}(\varepsilon)+\int_{\varepsilon}^x w(y)dy,$$
and $f(x)$ is the common density of the claims.

\begin{lemma}
For $w\in \mathcal{C}$ the function $Tw(x)$ is continuous on $[\varepsilon,K].$ 
Furthermore, the operator $T$ is Lipschitz with respect to the supremum norm on $\mathcal{C}.$
\end{lemma}

\begin{proof}
To prove continuity of $Tw(x)$ we show that for fixed $w\in \mathcal{C}$ the functions
$$x\to \{ M_{\gamma}(W)(x) -[c+rx+\theta x(\mu-r)]w(x)\} /(\theta^2\sigma^2x^2),\ \theta \in \mathcal{U},|\theta|\ge A(K)$$
are uniformly continuous on $[\varepsilon,K].$ This is true for the functions $$x\to [c+rx+\theta x(\mu-r)]w(x)$$ and 
$$x\to 1/(\theta^2\sigma^2x^2).$$
The representation
$$M_{\gamma}(W)(x)=\lambda \left(W(x)-\int_{x-\varepsilon}^x V_{\gamma}(x-y)f(y)dy-\int_0^{x-\varepsilon} W(x-y)f(y)dy\right)$$ 
shows that also $M_{\gamma}(W)(x)$ is continuous on $[\varepsilon,K].$
Since the maximum of uniformly continuous functions is continuous, we have continuity of $Tw(x).$

Now we consider two continuous functions $v(x),w(x)$ on $[\varepsilon,K]$ and use the norm 
$$||v(x)||=\sup\{|v(x)|:\varepsilon \le x\le K\}.$$ 
Then the inequalities
$$|V(x)-W(x)|\le \int_0^x |v(y)-w(y)|dy\le K ||v-w||,$$
$$\left| \int_0^{x-\varepsilon} (V(x-y)-W(x-y))f(y)dy \right|\le K^2||v-w||,$$
$$|[c+rx+\theta x(\mu-r)][v(x)-w(x)]|\le [c+rK+(a+b)K(\mu+r)|]|v-w||$$
together with boundedness of $1/(\theta^2x^2)$ imply that there exists a constant $C$
such that for all $v,w\in \mathcal{C}$ we have 
$$||Tv-Tw||\le C||v-w||.$$
\end{proof}

Using the standard Piccard-Lindel\"of argument we now obtain that for all $K>\varepsilon$ there exists
a continuously differentiable function $w(x)$ satisfying
\begin{equation}\label{DiffEq}
w'(x)=Tw(x),\ w(\varepsilon)=V_{\gamma}'(\varepsilon).
\end{equation}
With $K\to \infty$ we obtain a similar function defined on $[\varepsilon,\infty).$
We further show $w(x)>0$ on $[0,\infty)$.
Define $x_0=\inf\{x\ge 0: w(x)=0\}$.
Suppose $x_0<\infty$, then it holds $w(x_0)=0$. Thus
$$w'(x_0)=Tw(x_0)=
\inf_{\theta \in \mathcal{U},|\theta| > A(x_0)}M_\gamma(W)(x_0)/(\theta^2\sigma^2x_0^2)>0,$$
which contradicts:
$$w'(x_0)=\lim_{\epsilon\to 0+}\frac{w(x_0)-w(x_0-\epsilon)}{\epsilon}\le 0.$$
We then conclude $w$ is never $0$ and hence positive on $[0,\infty)$.

Next we proceed to select an appropriate function $A(x)$ such that an
anti-derivative of the solution $w$ of equation \eqref{DiffEq} solves the HJB equation.

For any $W\in C^1(0,\infty)$ and $W'(x)>0$, write
\begin{equation}
\label{fun1}
\begin{split}
\phi_W(x)&=\frac{2[M(W)(x)-(c+rx)W'(x)]}{(\mu-r)xW'(x)},\\
\psi_W(x)&=-\frac{(\mu-r)^2[W'(x)]^2}{2\sigma^2[M(W)(x)-(c+rx)W'(x)]}.
\end{split}
\end{equation}
One can check that $\psi_W(x)=-\frac{(\mu-r)W'(x)}{\sigma^2x\phi_W(x)}$, and that $\phi_W(x)=\alpha_W(x)$ and $\psi_W(x)=W''(x)$ 
if $W$ solves $\mathcal{L}^{(\alpha_W(x))}W(x)=0$.

\begin{lemma}
\label{S10}
\begin{itemize}
\item[(i)]{For $\mu>r$, since $V_a$  solves \eqref{HJBa}, 
it holds $\phi_{V_a(x)}\geq a$ on $(0, \varepsilon)$;}
\item[(ii)]{for $\mu<r$, since $V_{-b}$  solves \eqref{HJBb},
 it holds $\phi_{V_{-b}}(x)\leq -b$ on $(0, \varepsilon)$.}
\end{itemize}
\end{lemma}
Its proof is given in Appendix 1. 

We denote by $V^{(n)}$ the solution of equation \eqref{DiffEq} with
$A(x)\equiv \beta/n$ where $n$ is a positive integer and $\beta=\min\{a, b\}$.

Now define
$$x_n^*=\inf\{x:|\phi_{V^{(n)}}(x)|<\beta/n\},$$
and set $x_n^*=\infty$ if $\phi_{V^{(n)}}(x)\geq \beta /n$ for all $x\in[0,\infty)$.
We have $x_n^*>0$ by Lemma~\ref{S10}, and 
\begin{eqnarray}
\label{in5}
|\phi_{V^{(n)}}(x_n^*)|=\beta/n,
\end{eqnarray}
when $x_n^*<\infty$, by continuity of $\phi_{V^{(n)}}$.
Now we see $V^{(n)}$ is twice continuously differentiable and solves HJB equation~\eqref{HJB} on $(0,x_n^*)$;
further notice that for $x\in (0,x_n^*)$ we have $|\alpha^*_{V^{(n)}}(x)|>\beta/(n+1)$ (it equals $a$, $-b$ or $|\phi_{V^{(n)}}(x)|$).
Thus, $V^{(n)}(x)=V^{(n+1)}(x)$ for $x\in (0,x_n^*),$ and $x_n^*\leq x_{n+1}^*$ provided that both $x_n^*$ and $x_{n+1}^*$ are finite. 

Now we define 
\begin{equation}\label{A}
A(x)=\begin{cases}\beta &\ \ 0<x\le \varepsilon\\
\beta/(n+1) &\ \ x_{n}^*<x\le x_{n+1}^*  \end{cases},
\end{equation}
and denote
\begin{equation}
\label{xstar}
x^*=\underset{n\rightarrow \infty}{\lim}x_n^*;
\end{equation}
for any $x\in(0,x^*)$, write
\begin{equation}\label{V}
V(x)=\underset{n\rightarrow \infty}{\lim}V^{(n)}(x).
\end{equation}

We note it holds $V(x)=V^{(n)}(x)$ for $x\in [0,x^*_n]$.
From the previous discussions,
we see that $V$ is twice continuously differentiable and solves HJB equation~\eqref{HJB} on $(0,x^*)$.
And it holds $|\phi_V(x)|>A(x)$ on $(0,x^*)$. In the sequel, we write $v=V'$; and we have $v(x)>0$ on $(0,x^*)$. Further $V$ is bounded on $(0,x^*)$. In fact, one can show 
$V^{(n)}(\infty)$ is bounded (as the verification theorem) and $V^{(n)}(x)/V^{(n)}(\infty)$
is the maximal survival probability when investment control is restricted over region $\mathcal{U}_n:=[-b, -\beta/n]\cup [\beta/n, a]$. We then see that $V^{(n)}(\infty)$ decreases in $n$. 
This shows boundedness of $V$.

Now we proceed to show that $x^*$ is infinite. We state two lemmas without proof.
The following lemma implies that a function that solves the HJB equation coincides with
the fixed point of operator \eqref{Good Operator}.
\begin{lemma}
\label{UV}
Suppose function $u(x)$ satisfies 
\begin{itemize}
\item[(i)]{$u(x)=v(x)$ on $(0, x_0]$ and
$u(x)\in C^1(x_0,x_0+\delta_0)\cap C[x_0,x_0+\delta_0)$ where $x_0\geq 0$, 
$\delta_0>0$ and $x_0+\delta_0<x^*$;}
\item[(ii)]{function $U(x)$  
solves HJB equation \eqref{HJB} on $(x_0,x_0+\delta_0)$, where $U(x)=1+\int_0^x u(s)ds$.}
\end{itemize}
Then $u(x)=v(x)$ on $(x_0,x_0+\delta_0)$. 
\end{lemma}

For any $C^1$ function $W$, write
\begin{eqnarray}
\label{I}
I_W(x)=M(W)(x)-(c+rx)W'(x).
\end{eqnarray}
We have the following Lemma (we refer to \cite{S1} for its proof):
\begin{lemma}\label{Valphax0}
For $x_0>0$ with $I_V(x_0)>0$,
there exists a twice continuously differentiable function $V_{\alpha,x_0}$
satisfying:
\begin{itemize}
\item[(i)]{it is of the form
\begin{eqnarray}
\label{Valpha}
V_{\alpha,x_0}(x)=\begin{cases}
V(x)&\ \ 0\leq x\leq x_0\\
\int_{x_0}^xv_{\alpha,x_0}(s)ds+V(x_0)&\ \ x>x_0
\end{cases},
\end{eqnarray} 
where $v_{\alpha,x_0}\in C^1[x_0,\infty)$;}
\item[(ii)]{it solves
\begin{equation}\label{eq-alpha}
\mathcal{L}^{(\alpha_{\delta}(x))}\delta(x)=0,
\end{equation}
on $(x_0,\infty)$;}
\item[(iii)]{$V_{\alpha,x_0}$ is concave on $(x_0,\infty)$ and it holds $I_{V_{\alpha,x_0}}>0$ on $(x_0,\infty)$.}
\end{itemize}
\end{lemma}

Define sets
\begin{equation}
\label{S1S2}
\begin{split}
S_1=\{x:\phi_V(x)<a\},\\
S_2=\{x:\phi_V(x)>a\}.
\end{split}
\end{equation}

Now we obtain:
\begin{lemma}
\label{infty}
If $x_n^*<\infty$ for all $n\ge 1$, then
$\underset{n\rightarrow \infty}{\lim}x_n^*=\infty$.
\end{lemma}
\begin{proof}
 Assuming 
$\underset{n\rightarrow \infty}{\lim}x_n^*=x^*<\infty$, we prove Lemma~\ref{infty} by contradiction.

For the case $\mu-r>0$, 
define $x_1=\inf\{s:(s,x^*)\subset S_1\}$, where $S_1$ is defined in \eqref{S1S2}; 
from Lemma~\ref{VVa} and Lemma~\ref{S10},
we have $x_1>0$; now we show $x_1<x^*$ by contradiction. 
Suppose such $x_1$ does not exist. Then there exists a sequence
$x'_n\to x^*$ such that $\lim \phi_V(x'_n)\ge a$.
Define $y_n=\sup\{x: x<x^*_n, \phi_V(x)>a/2\}$,
then $\lim y_n=x^*$.
Further one can show $V=V_{\alpha,y_n}$ on $[y_n, x_n^*]$
with $\phi_V(y_n)=a/2$ and $\phi_V(x^*_n)=\beta/n$.
So there exits $z_n\in(y_n,x^*_n)$ such that
$$\lim_{n\to \infty}\phi'_V(z_n)=-\infty.$$
Notice $$M(V)'(x)=\lambda[V'(x)-f(x)-\int_0^xV'(x-y)f(y)dy].$$ 
There exists $K>0$ such that $M(V)'(x)>-K$ on $(0,x^*)$ due to boundedness of $f$ and $V$.
Further notice
\begin{equation}
\phi_V'(x)=\frac{2}{\mu-r}\frac{xM(V)'(x)V'(x)-M(V)(x)(V'(x)+xV''(x))+cV'^2(x)}{x^2V'^2(x)}
\end{equation}
thus
\begin{equation}
\begin{split}
\phi_V'(z_n)
&>\frac{2}{\mu-r}\frac{(-z_nK-M(V)(z_n))V'(z_n)-z_nM(V)(z_n)V''(z_n)}{z_n^2V'^2(z_n)}\\
&=\frac{2}{\mu-r}\frac{-z_nK-M(V)(z_n)-z_nM(V)(z_n)V''(z_n)/V'(z_n)}{z_n^2V'(z_n)},
\end{split}
\end{equation}
which tends to $-\infty$. From $V''(z_n)=V''_{\alpha,y_n}(z_n)<0$ and boundedness of $M(V)$,
it must hold $\lim V'(z_n)=0$ and then $\lim V''(z_n)=0$. 
Since $V$ solves the HJB equation, it then holds $\lim M(V)(z_n)=0$. Contradiction!
 
Thus $x_1$ exits such that $0<x_1<x^*$.  It holds $\phi_V(x_1)=a$ by
continuity of $\phi_V$. Hence we can select $x_0\in(x_1,x^*)$ such that $0<\phi_V(x_0)<a$;
then we have $I_V(x_0)>0$ and from Lemma~\ref{Valphax0},
there exists a function $V_{\alpha,x_0}$ which is
concave and twice continuously differentiable on $(x_0,\infty)$,
such that $V_{\alpha,x_0}(s)=V(s)$ on $(0,x_0]$, and it solves
$$\mathcal{L}^{(\alpha_{V_{\alpha,x_0}}(x))}V_{\alpha,x_0}(x)=0,$$
on $(x_0,\infty)$. 
Now define 
$$x_2=\sup\{x: x\geq x_0; 0<\phi_{V_{\alpha,x_0}}(t)<a,\forall\ t\in [x_0, x)\};$$
noticing $\mu-r>0$, $V''_{\alpha,x_0}(s)<0$, and
 $I_{V_{\alpha,x_0}}(s)>0$ (by Lemma~\ref{Valphax0}), and then $\phi_{V_{\alpha,x_0}(s)>0}$,
 it holds  
\begin{equation*}
0<\alpha_{V_{\alpha,x_0}}(s)=\phi_{V_{\alpha,x_0}}(s)<a,
\end{equation*}
for $s\in (x_0,x_2)$, we then have 
$$\underset{u\in[-b,a]}{\sup}\mathcal{L}^{(u)}V_{\alpha,x_0}(s)=
\mathcal{L}^{(\alpha_{V_{\alpha,x_0}}(s))}V_{\alpha,x_0}(s)=0,$$
which is the HJB equation (the maximizer of quadratic function
$\mathcal{L}^{(u)}V_{\alpha,x_0}(s)$ in $u$ 
is the vertex $\alpha_{V_{\alpha,x_0}}(s)$). By Lemma~\ref{UV}, we conclude that 
$V_{\alpha,x_0}(s)=V(s)$ and $\phi_V(s)=\phi_{V_{\alpha,x_0}}(s)<a$
 on $(x_0,x_2\wedge x^*)$. 
 If $x_2\in(x_0,x^*)$, by continuity of $\phi_{V_{\alpha,x_0}}$
 and definition of $x_2$, 
we have $\phi_V(x_2)=\phi_{V_{\alpha,x_0}}(x_2)=a$, which contradicts to $x_2\in S_1$.
If $x_2\geq x^*$, we have $\phi_{V_{\alpha,x_0}}(x^*)\neq 0$ since 
$I_{V_{\alpha,x_0}}(x^*)>0$ by Lemma~\ref{Valphax0}, and it contradicts to
$$\phi_{V_{\alpha,x_0}}(x^*)=\underset{n\rightarrow \infty}{\lim}\phi_{V_{\alpha,x_0}}(x_n^*)
=\underset{n\rightarrow \infty}{\lim}\phi_{V}(x_n^*)
=\underset{n\rightarrow \infty}{\lim}\phi_{V^{(n)}}(x_n^*)=0.$$
This finishes the proof for the case $\mu-r>0$.
The case $\mu-r<0$ can be shown in a similar way. 
\end{proof}

By Lemma~\ref{infty} and the previous discussions, we have the following theorem:
\begin{theorem} Function $V(x)$ defined in \eqref{V}, 
with the function $A(x)$ constructed in \eqref{A}, 
is a twice continuously differentiable solution to HJB equation~\eqref{HJB} on $(0,\infty)$
 with the initial values
$V(0+)=1$, $V'(0+)=\lambda/c$, and $V''_{\gamma}(0+)=\frac{\lambda}c \left( \frac{\lambda}c -F'(0+)-\frac{r+\gamma(\mu-r)}c \right)$.
\end{theorem}

In the following, we show some properties on the 
interplay between the optimal policies and parameters.
These theorems are stated via four parameter cases.
We give the proof of Theorem~\ref{case2} in
Appendix 2 and omit proofs the others which are similar.
\begin{theorem}\label{case1}
If $\mu>r$ and $a\geq b$, then 
\begin{itemize}
\item[(i)]{$V$ solves
HJB equation
\begin{equation}\label{HJBalpha}
\underset{\theta\in[-b,a]}{\sup}\mathcal{L}^{(\theta)}\delta(x)=\mathcal{L}^{(\alpha_{\delta}(x))}\delta(x)=0,
\end{equation}
on $S_1$ and HJB equation~\eqref{HJBa} on $S_2$;}
\item[(ii)]{the associated maximizer in the HJB equation is given by
\begin{eqnarray}\label{optc-1}
\alpha_V^*(x)=\begin{cases}\alpha_V(x)=\phi_V(x)&\ if\ \ \phi_V(x)<a\\
a&\ if\ \ \phi_V(x)\geq a\end{cases},
\end{eqnarray}
where it always holds $\phi_V(x)>0$.}
\end{itemize}
\end{theorem}

For $a\neq b$, define sets:
\begin{equation}
\label{S2ab}
\begin{split}
S_2^a&=\{x:a<\phi_V(x)<\frac{2ab}{b-a}\},\\
S_2^b&=\{x:\phi_V(x)>\frac{2ab}{b-a}\}.
\end{split}
\end{equation}
Then we have:
\begin{theorem}\label{case2}
If $\mu>r$ and $a<b$, then
\begin{itemize}
\item[(i)]{$V$ solves equation \eqref{HJBalpha}
on $S_1$, equation~\eqref{HJBa} on $S_2^a$, and equation~\eqref{HJBb} on $S_2^b$;}
\item[(ii)]{the associated maximizer in the HJB equation is given by 
\begin{eqnarray}\label{optc-2}
\alpha_V^*(x)=\begin{cases}\alpha_V(x)=\phi_V(x)&\ if\ \ \phi_V(x)<a\\
a&\ if\ \ a\leq \phi_V(x)\leq\frac{2ab}{b-a}\\
-b&\ if\ \ \phi_V(x)>\frac{2ab}{b-a}\end{cases},
\end{eqnarray}
where it always holds $\phi_V(x)>0$.}
\end{itemize}
\end{theorem}

Define sets
\begin{equation}
\label{S3S4}
\begin{split}
S_3&=\{x:\phi_V(x)>-b\},\\
S_4&=\{x:\phi_V(x)<-b\},\\
S_4^b&=\{x:-\frac{2ab}{a-b}<\phi_V(x)<-b\},\\
S_4^a&=\{x:\phi_V(x)<-\frac{2ab}{a-b}\}.
\end{split}
\end{equation}
We then obtain the following two theorems for the case $\mu<r$:
\begin{theorem}\label{case3}
If $\mu<r$ and $a\leq b$, then 
\begin{itemize}
\item[(i)]{$V$ solves the equation \eqref{HJBalpha}
on set $S_3$, equation~\eqref{HJBb} on set $S_4$;}
\item[(ii)]{the associated maximizer in the HJB equation  is given by
\begin{eqnarray}\label{optc-3}
\alpha_V^*(x)=\begin{cases}\alpha_V(x)=\phi_V(x)&\ if\ \ \phi_V(x)>-b\\
-b&\ if\ \ \phi_V(x)\leq -b\end{cases},
\end{eqnarray}
where $\phi_V(x)<0$.}
\end{itemize}
\end{theorem}

\begin{theorem}\label{case4}
If $\mu<r$ and $a>b$, then 
\begin{itemize}
\item[(i)]{$V$ solves the equation \eqref{HJBalpha}
on set $S_3$, equation~\eqref{HJBa} on set $S_4^a$, and equation~\eqref{HJBb} on $S_4^b$;}
\item[(ii)]{the associated maximizer in the HJB equation is given by 
\begin{eqnarray}\label{optc-4}
\alpha_V^*(x)=\begin{cases}\alpha_V(x)=\phi_V(x)&\ if\ \ \phi_V(x)>-b\\
-b&\ if\ \ \ -b\geq \phi_V(x)\geq-\frac{2ab}{a-b}\\
a&\ if\ \ \phi_V(x)<-\frac{2ab}{a-b}\\
\end{cases},
\end{eqnarray}
where $\phi_V(x)<0$.}
\end{itemize}
\end{theorem}

\begin{remark}
By the verification Theorem ~\ref{VER}, the function $V(x)$ is bounded and proportional to the maximal survival function $\delta$, in fact, $\delta(x)=V(x)/V(\infty)$.
\end{remark}

\begin{remark}
Noticing $|\alpha^*_V(x)|>A(x)$, the optimal policy always involves investment when $\mu\neq r$.
\end{remark}

\begin{remark}
In contrast to our case, when there are no constraints on the investment as in \cite{HP}, the optimal investment policy involves no shortselling of the risky asset although it is allowed.
\end{remark}

\begin{remark}\label{mu-r}
If $\mu=r$, the optimal investment strategy $\alpha_V^*$ equals $0$, $a$ or $-b$.
\end{remark}
\section{A Verification Theorem}
In this section, we prove a verification result; that is, we
show that the solution $V$ to the HJB equation is a multiple
of the maximal survival probability function.
The first lemma below shows that ruin is never caused by investment
if investment strategies are constant at low surplus levels.
\begin{lemma}\label{noruin}
For any non-negative integer $n$ and an admissible control policy $\pi$ such
that when $X^\pi_t$ is small $u_t\equiv a$ if $\mu>r$ or $u_t\equiv-b$ if $\mu<r$,
it holds $$P(\tau^\pi<\tau_{n+1}|\tau^\pi>\tau_n)=0,$$ where $\tau^\pi$ is defined in
~\eqref{ruintime}, and  $\tau_0=0$, $\tau_1$, $\tau_2$,... are
the times of claim arrivals.
\end{lemma}
Next we state an ergodicity result (Lemma~\ref{l1}) of the controlled surplus process
and non-triviality (Lemma~\ref{l2}) of the optimization.
For proofs of the lemmas we refer readers to \cite{AM} and \cite{S1}.
\begin{lemma}\label{l1}
For any admissible control policy $\pi$, the surplus process $X_t^\pi$ either
diverges to infinity or drops below $0$ with probability $1$.
\end{lemma}
\begin{lemma}\label{l2}
There exists a control policy $\pi$
 (e.g., a suitable constant investment strategy),
such that $P(\tau^\pi=\infty)>0$.
\end{lemma}
Now we prove the verification theorem:
\begin{theorem}\label{VER}
Suppose $g$ is a positive, increasing, and
twice continuously differentiable function on $[0,\infty)$;
and it solves the HJB equation~\eqref{HJB}. Then $g$ is bounded and
the maximal survival probability function is given by $\delta(x)=g(x)/g(\infty)$.
Moreover, the associated optimal investment strategy is $\pi^*=\{u^*(t)\}_{t\geq 0}$,
where $u^*(t)=\alpha_{\delta}^*(X^*_{t-})$, and $X^*_t$ is the surplus at time $t$ under
the control policy $\pi^*$.
\end{theorem}
\begin{proof}
For any $\epsilon>0$, we extend $g$ to $g_\epsilon$ such that $g_\epsilon$ is increasing and
twice continuously differentiable on $(-\infty,\infty)$ with $g_\epsilon(x)=0$ on
$(-\infty,-\epsilon)$ and $g_\epsilon(x)=g(x)$ on $[0,\infty)$.
For any admissible control $\pi$ and positive constant $M$,
define exit time
$$\tau^\pi_M=\inf\{t\geq 0: X^\pi_t\notin(0,M)\}.$$
Write $\tau^*_M=\tau^{\pi^*}_M$ to denote the first exit time from interval $(0,M)$
of the surplus process under the optimal control policy $\pi^*$.
By Ito's Lemma (See \cite{CT}), we have
\begin{equation}\label{eq2}
\begin{split}
g_\epsilon(X_{t\wedge\tau_M^*})
=&g(x)+\int_0^{t\wedge\tau_M^*}\{[c+rX_s^*+(\mu-r)\alpha_g^*(X^*_s)X^*_s]g'(X_s^*)\\
&+\frac{\sigma^2}{2}[\alpha_g^*(X_s^*)]^2(X_s^*)^2g''(X_s^*)\}ds
+\int_0^{t\wedge\tau_M^*}\sigma \alpha_g^*(X_s^*)g'(X_s^*)dW_s\\
&+\underset{s\leq t\wedge\tau_M^*,X_s^*<X_{s-}^*}{\sum}[g_\epsilon(X_s^*)-g(X_{s-}^*)]\\
=&g(x)+\int_0^{t\wedge\tau_M^*}\mathcal{L}^{(\alpha^*_{g}(X_s^*))}g(X_s^*)ds
+M^{(1)}_{t}+M^{(2)}_{t},
\end{split}
\end{equation}
where
\begin{eqnarray*}
M^{(1)}_t&=&\int_0^{t\wedge\tau^*_M}\sigma \alpha_g^*(X_s^*)g'(X_s^*)dB_s;\\
M^{(2)}_{t}&=&\underset{s\leq t\wedge\tau_M^*,X_s^*<X_{s-}^*}{\sum}[g_\epsilon(X_s^*)-g(X_{s-}^*)]\\
&&+\lambda\int_0^{t\wedge\tau_M^*}[g(X_s^*)-E(g(X_s^*-Y))]ds.
\end{eqnarray*}
Since $g$ solves the HJB equation, we have
$\mathcal{L}^{(\alpha^*_g(X_s^*))}g(X_s^*)=0$;
further noticing that $\alpha_g^*(X_s^*)g'(X_s^*)$
is bounded and  $M^{(1)}_t$, $M^{(2)}_t$ are both martingales,
we take expectation on both sides of \eqref{eq2} and obtain
\begin{eqnarray}\label{eq3}
E[g_\epsilon(X_{t\wedge\tau_M^*})]=g(x).
\end{eqnarray}
Similarly, for any admissible control $\pi=\{u_t\}_{t\geq 0}$, noticing
$\mathcal{L}^{(u_s)}g(X_s^\pi)\leq 0$, it holds
\begin{eqnarray}
E[g_\epsilon(X_{t\wedge\tau_M^\pi})]\leq g(x).
\end{eqnarray}
Letting $t\rightarrow\infty$ and then $M\rightarrow\infty$, by Fatou's Lemma, we obtain
\begin{eqnarray}
E[g_\epsilon(X_{\tau^\pi})]\leq g(x).
\end{eqnarray}
Further noticing by Lemma~\ref{l1} that
$$E[g_\epsilon(X_{\tau^\pi})]\geq g(\infty)P(\tau^\pi=\infty),$$
from Lemma~\ref{l2},
we must have that $g(\infty)$ is finite and that
\begin{eqnarray}\label{ineq1}
P(\tau^\pi=\infty)\leq\frac{g(x)}{g(\infty)}.
\end{eqnarray}
On the other hand, from \eqref{eq3}, we have
\begin{eqnarray*}
g(x)&=&E[g_\epsilon(X^*_{t\wedge\tau^M})]\\
&=& E(g(X^*_{t\wedge\tau^M})I_{\{X^*_{t\wedge\tau^M}>0\}})
+g(0)P(X^*_{t\wedge\tau^M}=0)\\
&+&E(g_\epsilon(X^*_{t\wedge\tau^M})I_{\{-\epsilon<X^*_{t\wedge\tau^M}<0\}})\\
&\leq& g(\infty)P(X^*_{t\wedge\tau^M}>0)+g(0)P(-\epsilon<X^*_{t\wedge\tau^M}<0),
\end{eqnarray*}
where in the last inequality we used $P(X^*_{t\wedge\tau^M}=0)=0$ which holds by Lemma~\ref{noruin}.
Notice $P(-\epsilon<X^*_{t\wedge\tau^M}<0)\rightarrow 0$ as $\epsilon\rightarrow 0$
since the claim distribution has a continuous density.
Letting $t\rightarrow\infty$ and then $M\rightarrow\infty$, we obtain
$$g(x)\leq g(\infty)P(\tau^*=\infty);$$
from \eqref{ineq1}, we then have $\frac{g(x)}{g(\infty)}=P(\tau^*=\infty)$, 
which is the maximal survival probability.
The proof is completed.
\end{proof}

\section{Analysis of the Case of Exponential Claims and Numerical Examples }

We will analyze  a specific case  in which  the claim size $Y$ has an exponential distribution with
$EY=m$ and the claims arrive with  intensity $\lambda$. We give several examples. 
We note that in Example 1, the parameters satisfy
\begin{equation}\label{cond}
m[{a}\mu +(1-{a})r-\lambda]+c<0,
\end{equation}
which is equivalent to $V_a''(0+)>0$ (all the parameters will be specified later).
This relation will ensure some nontrivial investment policy, which switches from
maximal long position to maximal shortselling and again to maximal long position when the
surplus increases. 
Before we introduce the examples we will state an auxiliary result needed for the analysis.

\begin{lemma}\label{linearinteg}
Let $F(x)=1 -e^{-x/m}$ be the distribution function of an exponential
random variable  with parameter $1/m$.  Then any bounded at 0 solutions to the integro-differential equation of the second order
\begin{equation}
\label{linear} \lambda\int_0^x V(x-z)\,dF(z)-\lambda
V(x) +V'(x)[\bar{\mu} x
+c]+\frac12\bar{\sigma}^2x^2V''(x)=0,
\end{equation}
with the initial condition
\begin{equation} \label{nachalo}
\lambda V(0)=cV'(0).
\end{equation}
is also a  solution to the linear differential equation of the third order
\begin{equation}\label{diffur}
x^{2}V^{\prime\prime\prime}(x)+ \left[\frac{x^2}{m} +
2\left(1+\frac{\bar{\mu}}{\bar{\sigma}^{2}}\right)\, x + \frac{2c}{\bar{\sigma}^{2}}
\right]V^{\prime\prime}(x)+
 2\left[\frac{\bar{\mu}x}{m\bar{\sigma}^2} + \frac{\bar{\mu} -\lambda
+c/m}{\bar{\sigma}^2}\right]V^\prime(x)=0.
\end{equation}
with the initial conditions (\ref{nachalo}) and  (\ref{diffur_cond}) below
\begin{equation}\label{diffur_cond}
\lim_{x\to +0} [c V^{\prime\prime}(x)+(\bar{\mu}-\lambda +c/m)
V^\prime (x)]=0.
\end{equation}
\end{lemma}
 \begin{proof}\thanks{The proof of this lemma was suggested to the authors by S.V. Kurochkin}.
Let $g(x)$ denote the left hand side of (\ref{linear}). Then differentiating it, we see that the left hand side of (\ref{diffur}) is equal to $g'(x)+g(x)/m$. Therefore if $V$ is a solution to
\eqref{linear} (that is $g(u)\equiv 0$) then it is also a solution to \eqref{diffur}.

From a general theory of the ordinary differential equations with singularities (e.g., see \cite{Fed},  \cite{Wazow} ) and a more detailed analysis in \cite{BelKonKur2} follows an existence of a two parametric family of solutions to (\ref{diffur}) bounded at 0 whose derivative is also  bounded at 0. Moreover each such bounded solution to (\ref{diffur}) satisfies (\ref{diffur_cond}) (cf. \eqref{V'V''0}).

Suppose $W(x)$ is any bounded at 0 solution to (\ref{diffur}), which is subject to (\ref{diffur_cond}).   The condition (\ref{diffur_cond}) implies that $W''(0+)$ is bounded as well.
Obviously for any $c_1,c_2$, the function $c_1W(x)+c_2$ is a bounded solution to (\ref{diffur}), satisfying (\ref{diffur_cond}).

Suppose $V(x)$ is a bounded at $0$  function which  satisfies (\ref{linear}) and (\ref{nachalo}). Let
$V(0)=\zeta, V'(0)=\nu$. Take $$\tilde{W}(x)=\nu(W(x)-W(0))/W'(0) +\zeta.$$ Then $\tilde{W}(x)$
obviously satisfies (\ref{diffur}), (\ref{diffur_cond}).  Substitute $\tilde{W}$ into   (\ref{linear}) instead of $V$, and let $\tilde{g}(x)$ denote the left hand side of  (\ref{linear}). Then the left hand side of (\ref{diffur}) is equal to $\tilde{g}'(x)+\tilde{g}(x)/m$. Thus
$\tilde{g}(x)=Ce^{-x/m}$ for some constant $C$.
Taking into account  that $\tilde{W}'(0)$ and $\tilde{W}''(0)$ are bounded, we can substitute
these expressions into the left hand side of (\ref{linear}) and see that $\tilde{g}(0)=
-\lambda \tilde{W}(0) +c\tilde{W}'(0)=0$ due to condition (\ref{nachalo}) and
the fact that $\tilde{W}(0)=V(0)$ and $\tilde{W}'(0)=V'(0)$ by construction.
Therefore $\tilde{W}$ satisfies (\ref{linear}) , (\ref{nachalo}).

From \cite{Kon3} and \cite{Kon2} we know that there exists at most one solution $V$  for an integro-differential equation (\ref{linear}) with given $V(0)$ and $V'(0)$ subject to (\ref{nachalo}). Thus $\tilde{W}(x)=V(x)$.
\end{proof}

\begin{lemma}\label{equationsol}
Let  $F(z)=1-e^{-z/m}$ be a distribution function of an exponential random variable with mean $m$. Then
the  integro-differential equation \eqref{linear}  with the boundary condition
\eqref{nachalo} has a  one-parametric family of solutions bounded at 0. In the neighborhood each solution of this family of $0$ is represented by an asymptotic series
\begin{equation}\label{asymp0}
 V(x)   = C_0 + D_1\left[x+\sum_{k=2}^{\infty}C_k x^k\right], \qquad x\to 0,
\end{equation}
where $C_0=V(0)$, $D_1=\frac{\lambda}{c}C_0$ and
\begin{equation}\label{tbel_asymp_phi}
C_k= D_k/k, \qquad k=2,3,\ldots,
\end{equation}
\begin{equation}\label{D_2}
D_2=-\left(\frac{\bar{\mu}-\lambda}{c}+\frac 1m\right), \quad
D_3=-\frac{D_2[\bar{\sigma}^2+2\bar{\mu}-\lambda +c/m]+\bar{\mu}/m}{2c},
\end{equation}
\begin{equation}\label{tbel_1D_k}
\begin{split}
D_k=&-\frac{D_{k-1}[(k-1)(k-2)\bar{\sigma}^2/2+(k-1)\bar{\mu}-\lambda+c/m]}{c(k-1)}\\
&+\frac{(1/m)D_{k-2}[(k-3)\bar{\sigma}^2/2+\bar{\mu}]}{c(k-1)},\qquad k\geq 4.
\end{split}
\end{equation}
The asymptotic expansion means that the difference between the $V(x)$ and the first $n$ terms in the right hand side of (\ref{asymp0}) is $o(x^n)$ when $x\to 0$.
\end{lemma}

The proof of this theorem can be obtained from Lemma \ref{linearinteg}, which reduces the integro-differential equation  (\ref{linear}) to an ordinary differential equation(\ref{diffur}), and a general theory of the  ordinary differential equations with pole-type singularities (see \cite{Wazow} Ch. 4 and \cite{KonNB}). The particular expression for the coefficients (\ref{tbel_asymp_phi})-(\ref{tbel_1D_k}) has been found in \cite{BelKonKur2}.

In the numerical examples, we consider exponential claim size distribution with
mean $m=1$. Further, set $\sigma=0.1$, $\lambda=0.09$, and $c=0.02$. 
The other parameters are given as follows via three examples. 

\begin{example}\label{ex1} ($\mu>r$ and $a<b$) We choose 
$\mu=0.02, a=1, b=20, r=0.015$. 
\end{example}

\begin{example}\label{ex2}($\mu<r$ and $a>b$) We choose 
$\mu=0.02, a=20, b=1, r=0.025$. 
\end{example}

\begin{example}\label{ex3} ($\mu<r$ and $a<b$) We choose 
$\mu=0.01, a=1, b=20, r=0.015$. 
\end{example}

Let $V$ be the solution of the HJB equation (\ref{HJB}) with $V(0)=1$.
From the previous sections, we know that $V(\infty)$ is finite and the function $V(x)/V(\infty)$ coincides with the maximal survival probability $\delta$ \eqref{opt-svp}; and the optimal  feedback control function $\alpha^*_V(x)$ is given by \eqref{alphastar}. If $\alpha _V^*(x)=\alpha _V(x)$
then the function $V$ satisfies
\begin{equation} \label{noconstrain} \lambda\int_0^x V(x-z)\,dF(z)-\lambda
V(x)+ (c+rx)V'(x) =\frac{1}{2}\frac{(\mu
-r)^2(V'(x))^2}{\sigma ^{2}V''(x)};\; \;
 \end{equation}
otherwise it satisfies (\ref{linear}) with
\begin{equation} \label{aa}
\bar{\mu}=a\mu +(1-a)r, \; \bar{\sigma}=a\sigma,
\end{equation}
when $\alpha _V^*(x)=a$, and with
\begin{equation} \label{ab}
\bar{\mu}=-b\mu +(1+b)r, \; \bar{\sigma}=b\sigma,
\end{equation}
when $\alpha _V^*(x)=-b$.

Numerical computations show that the optimal investment
strategy is quite surprising: there exists $0<x_1<x_2<x_3<\infty$ such that
\begin{eqnarray}
\alpha^*_V(x) = 
\begin{cases} a & 0\le x<x_1\\
 -b & x_1\le x<x_2\\
 a & x_2\le x\le x_3\\
 \alpha_V(x) & x>x_3.
\end{cases}
\end{eqnarray}
In the following we show that  such investment strategies typically occur
when (5.1) holds and $b$ is large.

From Lemma \ref{VVa} we know that in the neighborhood of 0 the function $V$ satisfies
(\ref{VVa}), that is (\ref{linear}) with $\bar{\mu}$ given by (\ref{aa}); and $\alpha _V^*(x)=a$ in
this neighborhood. In view of (\ref{cond}), the coefficient $D_2$ in the asymptotic
expansion (\ref{asymp0}) of the function $V$ at 0 is positive. Thus $V^{\prime\prime}(x)\sim D_1D_2>0 $, when $x\to 0$ and
\begin{equation}
\label{sim}
\alpha _V(x)\sim -\frac{(\mu -r)}{\sigma ^2 D_2x}<0.
\end{equation}
If one chooses instead of $\alpha _V(x)$ its asymptotic approximation given by (\ref{sim}), then
we see that $-\frac{(\mu -r)}{\sigma ^2 D_2x}> \frac{a-b}{2}$ for
$x>{\underline{x}}=2\frac{(\mu-r)}{\sigma ^2}\frac{1}{D_2(b-a)}$. If $b$ is large enough then
$\underline{x}$ is small and the ratio of $\alpha_V(x)$ and the right hand side of (\ref{sim}) is close to 1 (note that for $\bar{\mu}$ given by (\ref{aa}) the asymptotic expansion (\ref{asymp0}) does not depend on $b$) and  $\alpha_V(\underline{x})> \frac{{a }-b}{2}$. The choice of $a$ or $b$ in our example ensures that this is the case, as numerical
computations show. Therefore there exists
$x_1>0$ such that $\alpha _V(x) > \frac{{a }-b}{2}$ for $x>x_1$, while $V''(x_1)>0$ and
$\alpha_V(x)<\frac{{a}-b}{2}$ for $x<x_1$. In view of \eqref{alphastar}, the function $V$ satisfies (\ref{linear}) with $\bar{\mu}$ given by (\ref{ab}) in the right neighborhood of  $x_1$ and it satisfies (\ref{linear}) with $\bar{\mu}$ given by  (\ref{aa}) for $x<x_1$.

This shows that $\alpha _V^*(x)=a$ for $x<x_1$ and $\alpha _V^*(x)=-b$ in the right neighborhood of $x_1$.  Let us show that there exists $x_2$ such that $\alpha _V^*(x)=-b$ for
$x_1<x<x_2$, while $\alpha _V^*(x)=a$ in the right neighborhood of  $x_2$. That is the point $x_1$ is the first point of the ``extreme" switching from the maximal long position to the maximal short position in  the risky asset, while $x_2$ is the point of the second ``extreme" switching from the maximal short to the maximal long position.  Since $V''(x_1)>0$ the function $V$ is
convex in the neighborhood of $x_1$. On the other hand, if $V''(x)\ge 0$ for all $x>x_1$ then
$V'(x)\ge V'(x_1)>0$ for all $x>x_1$ and $V(x)\to\infty$ as $x\to\infty$. Contradiction! 
Therefore there exists a point $\bar{x}$ such that $V''(\bar{x})=0$ while
$V''(x)>0$ for $x<\bar{x}$. Therefore, $\alpha_V(x)\to-\infty$ as $x\uparrow \bar{x}$ and there
exists $x_2$ such that $\alpha_V(x)>\frac{a-b}{2}$ for $x_1<x<x_2$ 
and $\alpha _V(x) \leq \frac{{a }-b}{2}$ in the right neighborhood of $x_2$.

Suppose that $\alpha^*_V(x)=a$ for all $x>x_2$. Then the function $V$ satisfies \eqref{linear}
with $\bar{\mu}$ and $\bar{\sigma}$  given by (\ref{aa}).  From the general theory of the differential equations with singularity at infinity (see \cite{Wazow} Chapter 4, or \cite{Fed}) follows that for a solution $V$  of  (\ref{linear}) we have
 $$V'(x)=dx^{-2\bar{\mu}/\bar{\sigma}^2}[1+o(1)],$$ and
   $$V ''(x)=d(-2\bar{\mu}/\bar{\sigma}^2)x^{-2\bar{\mu}/\bar{\sigma}^2-1}[1+o(1)],$$
 when $x\rightarrow \infty$ for some  positive $d$.
(This precise formula on asymptotic approximation for
$V'$ and $V''$ was  proved in \cite{BelKonKur2}, using an
approach developed in \cite{KonNB}. Also a similar result was obtained in \cite{FrolKP}). Thus for $V$ being a solution to (\ref{linear}), it holds
\begin{equation}
\label{infty1}
\alpha_V(x)\to\frac{(\mu -r)\bar{\sigma}^2}{2\bar{\mu}\sigma
^2 }, \quad x\rightarrow \infty.
\end{equation}
Obviously the right hand side of \eqref{infty1} is positive and less than $a$. Really
$\frac{(\mu -r)\bar{\sigma}^2}{2\bar{\mu}\sigma^2 }<a$ is equivalent to
$\frac{1}{a}+\frac{r}{(\mu-r)a^2}>\frac{1}{2a}$  which is obvious. Thus there exists
$x_3$ such that $\alpha_V^*(x)=a$ for $x_2<x<x_3$ and $\alpha_V^*(x)=\alpha_V(x)$ for
$x>x_3$. That is the switching at the point $x_3$ is from the maximal long position to a long
position which is less than the maximal possible.
The numerical computations are depicted in the graphs. 
(The authors would like to thank Yuri Gribov for providing numerical  
computations for the first two examples.)

\begin{figure}
\begin{center}
\includegraphics[width=3.5 in]{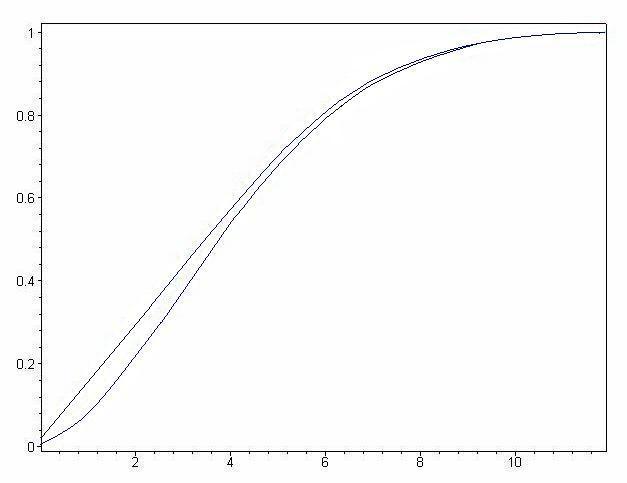}
\caption{Example 1 - Maximal survival probability 
(The lower curve is for the no-borrowing-no-shortselling case)}
\end{center} 
\label{ex1-maxprob}
\end{figure}

\begin{figure}
\begin{center}
\includegraphics[width=3.5 in]{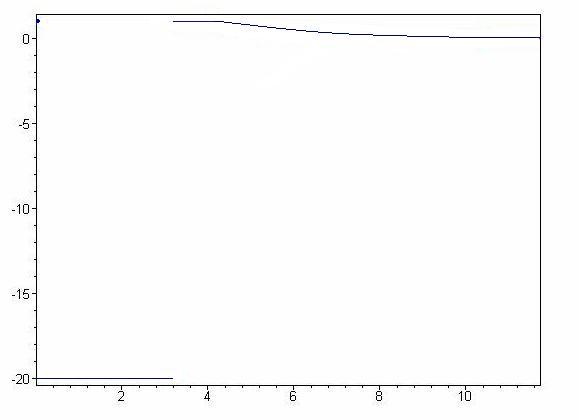}
\caption{Example 1 - Optimal investment proportion}
\end{center} 
\label{ex1-optinv}
\end{figure}

\begin{figure}
\begin{center}
\includegraphics[width=3.5 in]{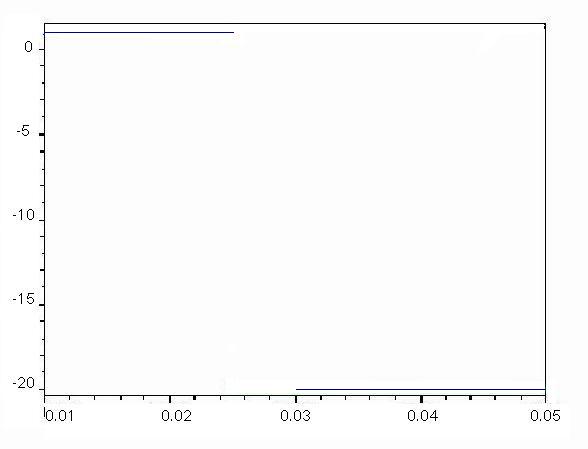}
\caption{Example 1 - Optimal investment proportion near zero surplus}
\end{center} 
\label{ex1-nearzero}
\end{figure}

\begin{figure}
\begin{center}
\includegraphics[width=3.5 in]{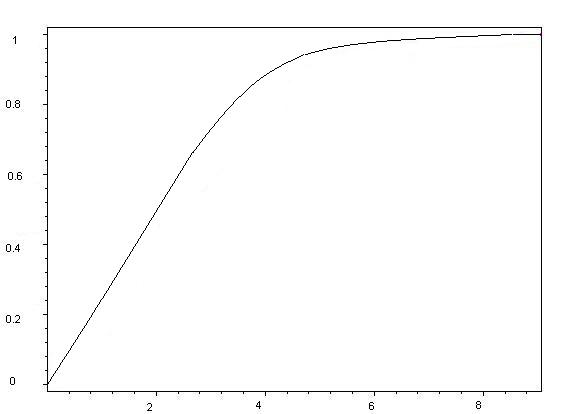}
\caption{Example 2 - Maximal survival probability}
\end{center} 
\label{ex2-maxprob}
\end{figure}

\begin{figure}
\begin{center}
\includegraphics[width=3.5 in]{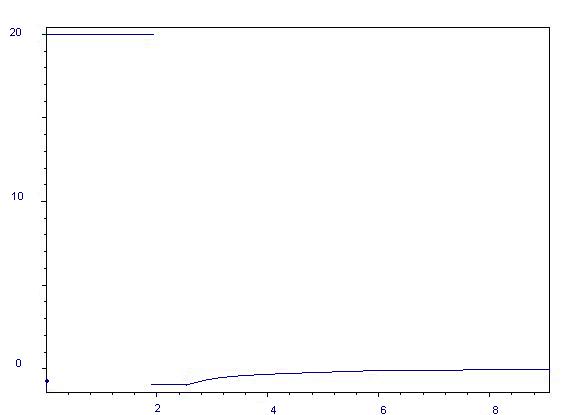}
\caption{Example 2 - Optimal investment proportion}
\end{center} 
\label{ex2-optinv}
\end{figure}

\begin{figure}
\begin{center}
\includegraphics[width=3.5 in]{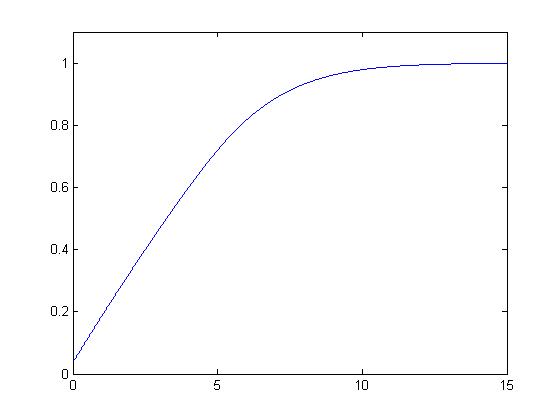}
\caption{Example 3 - Maximal Survival Probability}
\end{center} 
\label{ex3-maxprob}
\end{figure}

\begin{figure}
\begin{center}
\includegraphics[width=3.5 in]{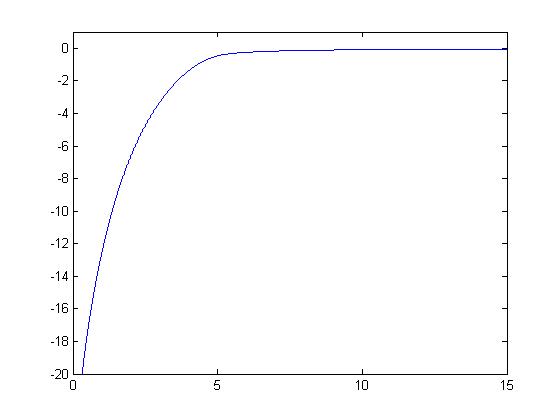}
\caption{Example 3 - Optimal investment proportion}
\end{center} 
\label{ex3-optinv}
\end{figure}

\section{Economic Interpretation}
As we can see from Lemma~\ref{VVa} and Lemma~\ref{VVb} in the neighborhood of 0, when the surplus level is low, the optimal policy is always to take position which brings the highest rate of return. If $\mu>r$ then it is optimal to borrow and to take the maximal long position, if
$\mu<r$ then the optimal policy requires to shortsell the risky asset and put into the risk-free asset the maximal allowable amount.  In case $a>b$ and $\mu>r$  the optimal policy would require no
shortselling, maintaining at low levels maximal long position and then purchasing risky asset at the level less than the maximum possible so as to reduce the volatility of the portfolio. Like wise when $\mu<r$ and $b>a$,  only shortselling is needed, at first at the maximal possible level and then at the level lower than the maximal.

The nature of the optimal policy becomes less obvious, when the surplus level increases in the case when $b>a$, while $\mu>r$  .  Depending on the available actions (that is, depending on $a$ and $b$) we might have qualitatively  completely  different optimal policies.  In this case,  the nature of the optimal policy is the result of a rather nontrivial  interplay between  the  rate of return  and the volatility.  As we saw in the example analyzed  in Section 6, when the condition (\ref{cond}) is
satisfied and $b$ is much larger than $a$, then at certain surplus levels we have to switch from the maximal long position, which was used at low surplus levels to the maximal short position, now ``gambling" on the effect of a largest possible  volatility which must increase the surplus level  with higher probability than otherwise would be the case.  At first glance this policy might appear counterintuitive, if one takes into account that for such a policy the rate of return is the lowest possible.  When the surplus level increases even more, it becomes again optimal to stick to the highest possible rate of return; and with high surplus levels, it is optimal to have lower than the maximal possible rate of return, simultaneously having lower volatility. It is worth mentioning, that when $b$ is not much larger than $a$ a similar analysis show that the effect of switching to the short position is not observed, and no shortselling is optimal at any surplus levels; i.e., the set $S_2^b$ defined in \eqref{S2ab} could be empty then.

A similar phenomenon can be observed when  $\mu<r$ and $a>b$. For certain values of the parameters, the optimal policy at low surplus levels would involve shortselling the risky asset so as to have the maximal rate of return for the resulting portfolio, while at higher surplus level we might observe a switch to the policy with the highest volatility, that is to the one with maximal long position in the risky asset even though that is the policy with the lowest return rate.

It is worth mentioning that none of those effects can be observed when we have either unconstrained case as in  \cite{HP}  or the no-borrowing-no-shortselling case  of \cite{AM}. In \cite{HP} the ability to have unlimited large leverage enables one to adhere only to the policies without shortselling,
 while in \cite{AM}, the no-shortselling constraint does not allow to achieve sufficiently large volatility, of the portfolio, so as to reach higher levels with higher probability while having a lower rate of return.


{\bf Appendix 1.}\\
We now show Lemmma~\ref{S10} and write
\begin{equation}
\label{fun2}
\begin{split}
\xi_{\gamma,W}(x)&=-\frac{\gamma^2(\mu-r)xW'(x)}{2[M(W)(x)-(c+rx+(\mu-r)\gamma x)W'(x)]},\\
\eta_{\gamma,W}(x)&=\frac{2[M(W)(x)-(c+rx+(\mu-r)\gamma x)W'(x)]}{\sigma^2\gamma^2x^2},
\end{split}
\end{equation}
for any function $W$ with $W'>0$.
Note that $\eta_{\gamma,W}(x)=-\frac{(\mu-r)W'(x)}{\sigma^2x\xi_{\gamma,W}(x)}$, and that $\xi_W(x)=\alpha_W(x)$ and $\eta_{\gamma,W}(x)=W''(x)$ 
if $W$ solves $\mathcal{L}^{(\gamma)}W(x)=0$.

We prove result $(i)$ by contradiction.
For convenience, note $V=V_a$ on $(0,\varepsilon)$ when $\mu>r$.
Suppose $\phi_V(x)<a$ for some $x\in(0,\varepsilon)$.
From $\mathcal{L}^{(a)}V(x)=0$, we have
 $\xi_{a,V}(x)=\alpha_V(x)$ and $\eta_{a,V}(x)=V''(x)$.
From $\phi_V(x)<a$ and $\mu-r>0$, we have
$$2[M(V)(x)-(c+rx)V'(x)]<a(\mu-r)xV'(x),$$
wherefrom it holds 
\begin{eqnarray}
\label{in1}
2[M(V)(x)-(c+rx+(\mu-r)ax)V'(x)]<-a(\mu-r)xV'(x)<0,
\end{eqnarray}
and it implies
$V''(x)=\eta_{a,V}(x)<0$; further notice that $V$ solves
 $$\underset{\theta\in[-b,a]}{\sup}\mathcal{L}^{(\theta)}V(x)=\mathcal{L}^{(a)}V(x)=0,$$
where $a$ is the maximizer. Hence the vertex
of the quadratic function $\mathcal{L}^{(\theta)}V(x)$ in $\theta$ must be to the right of $a$,
i.e., it must hold $\alpha_V(x)=\xi_{a,V}(x)\geq a$; with \eqref{in1}, we obtain
\begin{eqnarray}
\label{in2}
2a[M(V)(x)-(c+rx+(\mu-r)ax)V'(x)]\geq -a^2(\mu-r)xV'(x).
\end{eqnarray}
Inequality \eqref{in2} yields
$\phi_V(x)\geq a$. Contradiction!
Result $(ii)$ can be proved in the same manner.

\vspace{.5cm}
To show Theorem~\ref{case2}.
we need the following Lemma:
\begin{lemma}
\label{Vgammax0}
For  $x_0\geq 0$ and $\gamma\in\{a, -b\}$, there exists function $V_{\gamma,x_0}$ 
satisfying the following: $V_{\gamma,x_0}(x)=V(x)$ on $[0,x_0]$; $V'_{\gamma,x_0}(x_0)=v(x_0)$;
and for $x\in (x_0,\infty)$
\begin{equation}\label{eq-gamma}
\mathcal{L}^{(\gamma)}V_{\gamma,x_0}(x)=0. 
\end{equation}
\end{lemma}

\vspace{.5cm}

{\bf Appendix 2.}

Now we prove Theorem~\ref{case2}.
For any $x\in S_1$, define $x_1=\inf\{s:(s,x)\subset S_1\}$. From Lemma~\ref{S10},
we have $x_1>0$. It holds $\phi_V(x_1)=a$ by continuity. 
Select $x_0\in(x_1,x)$ such that $A(x_0)<\phi_V(x_0)<a$.
Then it holds $I_V(x_0)>0$. 
From Lemma~\ref{Valphax0}, there exists a function $V_{\alpha,x_0}$ 
concave and twice continuously differentiable on $(x_0,\infty)$ that solves
$$\mathcal{L}^{(\alpha_{V_{\alpha,x_0}}(x))}V_{\alpha,x_0}(x)=0,$$
on $(x_0,\infty)$. 
Now define 
$$x_2=\sup\{s:s\geq x_0;A(t)<\phi_{V_{\alpha,x_0}}(t)<a,\forall\ t\in [x_0, s)\};$$
noticing $V''_{\alpha,x_0}(s)<0$ and  
\begin{equation*}
0<\alpha_{V_{\alpha,x_0}}(s)=\phi_{V_{\alpha,x_0}}(s)<a,
\end{equation*}
for $s\in (x_0,x_2)$, we have 
$$\underset{\theta\in[-b,a]}{\sup}\mathcal{L}^{(\theta)}V_{\alpha,x_0}(s)=
\mathcal{L}^{(\alpha_{V_{\alpha,x_0}}(s))}V_{\alpha,x_0}(s)=0,$$
i.e., $V_{\alpha,x_0}$ solves the HJB equation. By Lemma~\ref{UV}, we conclude
$V_{\alpha,x_0}(s)=V(s)$ and $\phi_V(s)=\phi_{V_{\alpha,x_0}}(s)<a$
 on $(x_0,x_2)$. Since $x_2$ is the supremum, if $x_2\leq x$, then we have $\phi_V(x_2)=a$ 
 which contradicts $x_2\in S_1$. Thus it holds $x_2>x$ and we conclude that
 $V$ equals $V_{\alpha,x_0}$ in a neighborhood of $x$ and solves \eqref{HJBalpha} with maximizer $\alpha^*_V=\alpha_V$.
 
For any $x\in S_2^a$, 
we choose $x_0(<x)$ such that $2a<\phi_V(s)<\frac{2ab}{b-a}$ for all $s$ on $[x_0,x]$;
from Lemma~\ref{Vgammax0}, there exists $V_{a,x_0}$ such that $V_{a,x_0}(s)=V(s)$
for $s\in [0,x_0]$ and it solves $\mathcal{L}^{(a)}V_{a,x_0}(s)=0$ on $(x_0,\infty)$. 
Hence we have $V''_{a,x_0}(s)=\eta_{a,V_{a,x_0}}(s)$ and $\alpha_{V_{a,x_0}}(s)=\xi_{a,V_{a,x_0}}(s)$ on $(x_0,\infty)$.
Now define 
 $$x_2=\sup\{s:s\geq x_0; 2a<\phi_{V_{a,x_0}}(t)<\frac{2ab}{b-a},
\forall\ t\in [x_0, s)\};$$
it follows $V''_{a,x_0}(s)=\eta_{a,V_{a,x_0}}(s)>0$ and 
$\alpha_{V_{a,x_0}}(s)=\xi_{a,V_{a,x_0}}(s)<\frac{a-b}{2}$. 
Thus $V_{a,x_0}$ solves HJB equation~\eqref{HJBa} on $(x_0,x_2)$. 
By Lemma~\ref{UV}, we conclude 
$V_{a,x_0}(s)=V(s)$ and $\phi_{V_{a,x_0}}(s)=\phi_V(s)$
 on $(x_0,x_2)$; since $x_2$ is the supremum, if $x_2\leq x$, 
then $\phi_V(x_2)=2a$ or $\phi_V(x_2)=\frac{2ab}{b-a}$.
This contradicts $x_2\in S^a_2$. 
Hence we have $x<x_2$. Thus $V$ is equal to $V_{a,x_0}$ and solves 
\eqref{HJBa} with maximizer $\alpha^*_V=a$ on $(x_0,x_2)$ that contains $x$.

Similarly one can show for any  $x\in S_2^b$, in a neighborhood of $x$, 
 $V$ is equal to $V_{-b,x_0}$ and solves \eqref{HJBb} with maximizer $\alpha^*_V=-b$.

\begin{thebibliography}{999999999}

\addcontentsline{toc}{chapter}{reference}
\bibitem{AM} Azcue, P. and Muler, M.:
    Optimal Investment Strategy to Minimize the Ruin Probability of an Insurance Company under Borrowing Constraints,
        {\it Insurance Math. Econom.} {\bf 44}(1) 26--34, (2009)


\bibitem{BelKonKur2} Belkina, T.A.,  Konyukhova N.B.,  Kurkina, A. O.: Optimal control of investments in dynamical insurance models: II. Cramer-Lundberg model with exponential claim size distribution (in Russian), {\it Survey of the Industrial and Applied Mathematics} 17, 3--24 (2010).

\bibitem{B} Browne, S.:
    Optimal investment policies for a firm with a random risk process: exponential utility and minimizing the probability of ruin,
        {\it Math. Ope. Res} {\bf 20}(4), 937--958 (1995)

\bibitem{CT} Cont, R. and Tankov, P.:
    Financial modelling with jumps processes,
      Chapman \& Hall/CRC,  Boca Raton, (2004)


\bibitem{HP} Hipp, C. and Plum, M.:
    Optimal investment for insurers,
        {\it Insurance Math. Econom.} {\bf 27}, 215--228 (2000)

\bibitem{HP2} Hipp, C. and Plum, M.:
    Optimal investment for investors with state dependent income, and for insurers. {\it Finance and Stochastics} (3) {\bf 7}, 299--321 (2003)


\bibitem{Fed} Fedoryuk, M.V: {\it Asymptotic analysis: Linear
      Ordinary Differential Equations}. Springer, Berlin, 1993.

\bibitem{FrolKP} Frolova, A., Kabanov, Y., and Pergamenshchikov S.:  In the insurance
       business risky investments are dangerous, {\it  Finance and Stochastics},
       6,  227-235 (2002).

\bibitem{KonNB} Konyukhova, N.B.:  Singular Cauchy problems
        for systems of ordinary differential equations, {\it U.S.S.R.
        Comput. Maths. Math. Phys.} , 23,  72--82 (1983).

\bibitem{Kon3} Konyukhova, N.B.: Singular Cauchy problems for some systems of nonlinear functional differential equations.
{\it Differential Equations}, 31,  1286-1293 (1995).

\bibitem{Kon2} Konyukhova, N.B.: Singular problems for systems of nonlinear functional differential equations. {\it Int. Scientific Journal Spectral and Evolution Problems},  20, 199--214 (2010)

\bibitem{LUO} Luo, S.:
    Ruin minimization for insurers with borrowing constraints,
     {\it North American Actuarial Journal}, {\bf 12} (2), 143 -- 174 (2008)

\bibitem{PY} Promislow, S.D. and Young, V.:
    Minimizing the probability of ruin when claims follow Brownian motion with drift,
       {\it North American Actuarial Journal} {\bf 9}(3), 109--128 (2005)

\bibitem{S} Schmidli, H.:
    Optimal proportional reinsurance policies in a dynamic setting,
        {\it Scan. Actuarial J.} {\bf 1}, 55--68 (2001)

\bibitem{S1} Schmidli, H.:
    On minimizing the ruin probability by investment and reinsurance,
        {\it The Annals of Applied Probability} {\bf 12(3)}, 890--907 (2002)

\bibitem{TM} Taksar, M. and Markussen, C.:
    Optimal dynamic reinsurance policies for large insurance portfolios,
        {\it Finance and Stochastics} {\bf 7} 97--121 (2003)

\bibitem{Wazow}  Wasow, W:  {\it  Asymptotic Expansion for Ordinary  Differential Equations}. Dover, New York, 1987.

\end{thebibliography}
\end{document}